\documentclass[final,twoside,11pt]{entics} 
\usepackage{enticsmacro}

\usepackage{amsmath}

\usepackage{amssymb}
\usepackage{ifthen}
\usepackage{xspace}
\usepackage{xparse}
\usepackage{mathpartir}
\usepackage{cleveref}
\usepackage{enumitem}
\usepackage{stmaryrd}
\usepackage{multicol}
\usepackage{tikz-cd}
\usepackage{mathtools}
\usepackage{stackengine}
\usepackage{scalerel}
\usepackage{wasysym}
\usepackage{fontawesome}
\usepackage{xcolor}
\usepackage{tcolorbox}
\usepackage{accents}

\usetikzlibrary{matrix}
\usetikzlibrary{calc}
\usetikzlibrary{positioning}
\usetikzlibrary{decorations.pathreplacing}

\DeclareFontFamily{OT1}{pzc}{}
\DeclareFontShape{OT1}{pzc}{m}{it}{<-> s * [1.10] pzcmi7t}{}
\DeclareMathAlphabet{\mathpzc}{OT1}{pzc}{m}{it}

\crefname{figure}{fig.}{fig.}
\Crefname{figure}{Fig.}{Fig.}
\Crefname{section}{Sec.}{Sec.}

\definecolor{light-gray}{gray}{0.85}

\input{macro}

\def\conf{MFPS 2022} 	%%Fill in the acronym for your conference (with year)
\volume{1}			%Fill in the ENTICS volume number here
\def\volu{1}			% and here
\def\papno{7}			%Fill in your paper number here
\def\mydoi{{\normalshape \href{https://doi.org/10.46298/entics.10360}{10.46298/entics.10360}}}
\def\copynames{J. Z. S. Hu, B. Pientke} %% Fill in the first initial and last name of the authors
\def\runtitle{Categorical Normalization Proof for Modal Lambda-Calculus}

\def\lastname{Hu and Pientka}
\begin{document}
%%% Note the beginning and end of the frontmatter section that starts here%%%%%

\begin{frontmatter}
  \title{A Categorical Normalization Proof for the Modal Lambda-Calculus}

  \author{Jason Z. S. Hu\thanksref{jasonemail}}	%%Note NO SPACE between
  \and
  \author{Brigitte Pientka\thanksref{brigitteemail}}
  \address{School of Computer Science\\ McGill University\\			%last name and \thanksref{...} 
    Montr\'eal, QC, Canada}  							%or between \thanksrefs...
  \thanks[jasonemail]{Email: \href{mailto:zhong.s.hu@mail.mcgill.ca} {\texttt{\normalshape
        zhong.s.hu@mail.mcgill.ca}}}
  
  %%% Note: if both authors share same institution, the only list the address once, after the second 
  %%% author.
  %%% There also is a link from the first author to the co-author's address to show how to list 
  %%% affiliations to more than one institution, when needed. 
  \thanks[brigitteemail]{Email: \href{mailto:bpientka@cs.mcgill.ca} {\texttt{\normalshape
        bpientka@cs.mcgill.ca}}}

  \begin{abstract}
    We investigate a simply typed modal $\lambda$-calculus, \lambox, due to Pfenning, Wong
    and Davies, where we define
    a well-typed term with respect to a context stack that captures the possible world
    semantics in a syntactic way. It provides logical foundation for
    multi-staged meta-programming. Our main contribution in this paper is a
    normalization by evaluation (NbE) algorithm for \lambox which we prove sound and
    complete. The NbE algorithm is a moderate extension to the standard presheaf model
    of simply typed $\lambda$-calculus. However, central to the model construction and
    the NbE algorithm is the observation of \emph{Kripke-style
      substitutions} on context stacks which brings together two previously separate
    concepts, structural modal transformations on context stacks and substitutions for
    individual assumptions. Moreover, Kripke-style substitutions allow us to give a
    formulation for contextual types, which can represent open code in a
    meta-programming setting. Our work lays the foundation for
    extending the logical foundation by Pfenning, Wong, and Davies
    towards building a practical, dependently typed foundation for
    meta-programming. 
  \end{abstract}
  
\end{frontmatter}

\section{Introduction}\label{sec:intro}
The Curry-Howard correspondence fundamentally connects formulas and
proofs to types and programs. This view not only
provides logical explanations for computational
phenomena, but also serves as a guiding principle in designing type
theories and programming languages. 

Extending the Curry-Howard correspondence to modal logic has been fraught with
challenges. One of the first such calculi for the modal logic S4 were proposed by
Bierman and de Paiva~\cite{bierman_intuitionistic_1996,bierman_intuitionistic_2000}
% which however still was unsatisfying from a proof-theoretic view as
% it required many commuting conversion and it is unsatisfactory from
% a proof-theory point of view: the connective box modality that we
% are introducing in the conclusion, is being used in the premises
% \cite{Martini:96,Bierman:96,DBLP:conf/icalp/GhaniPR98}
and subsequently by Pfenning and Davies~\cite{pfenning_judgmental_2001}. A key
characteristic of this work is to separate the assumptions that are valid in every
world from the assumptions that presently hold in the current world. This leads to a
dual-context style formulation that satisfies substitution properties (see for example
\cite{DBLP:conf/icalp/GhaniPR98}).
% However, the elimination form of the
% box-modality also introduces permutation conversions which make it
% challenging to extend to dependent type theories and prove strong normalization. In fact, usually this
% property is not proven directly, but via a translation to another % intuitionistic
% logic.
% where we face a similar issue of
%permutation-conversions and strong normalization proof already exist
%by \cite{Groote:RTA99}. 

% The three pillars of the Curry-Howard correspondence:
% confluence, strong normalization, and the sub-formula property due to
% Prawitz.
% NAMELY:
% If N:A is a normal term N of type A (i.e. there is no reduction N -> N')
% then the derivation of Γ |- N : A can only mention subformulas of A and
% subformulas of assumptions in Γ

In recent years, modal type systems based on this
dual-context style have received renewed attention and provided
insights into a wide range of seemingly unconnected areas: from
reasoning about universes in homotopy type
theory~\cite{licata_internal_2018,shulman_brouwers_2018} to 
mechanizing meta-theory~\cite{Pientka:POPL08,Pientka:LICS19}, to reasoning about
effects~\cite{Zyuzin:ICFP21}, and meta-programming 
\cite{Jang:POPL22}. This line of work builds on the dual-context
formulation of Pfenning and Davies. However, due to the permutation
conversions it is also challenging to extend to dependent type
theories and directly prove normalization via logical relations. 
An alternative to dual-context-style modal calculi is pursued by Clouston,
Birkedal and collaborators (see
\cite{clouston_fitch-style_2018,gratzer_implementing_2019}).
This line of work is inspired
by  the Fitch-style proof representation given by Borghuis~\cite{borghuis_coming_1994}. Following Borghuis they have 
been calling their representation the Fitch style.
Fitch-style systems model Kripke semantics~\cite{kripke_semantical_1963} and use
 locks to manage assumptions in a context.  To date, existing
formulations of $S4$ in Fitch style~\cite{clouston_fitch-style_2018,gratzer_implementing_2019}
mainly considers idempotency where $\square T$ is isomorphic to
$\square\square T$. However, this distinction is important from 
a computational view. For example, in multi-staged programming (see
\cite{pfenning_judgmental_2001,davies_modal_2001}) $\square T$
describes code generated in one stage, while $\square\square T$
denotes code generated in two stages. It is also fruitful to keep the
distinction from a theoretical point of view, as it allows for a fine
grained study of different, related modal logics. 

% \subsection{A road map}.
In this paper, we take \lambox, an intuitionistic version of modal logic $S4$, from
Pfenning, Wong and Davies~\cite{Pfenning95mfps,davies_modal_2001} as a starting
point.
Historically, \lambox is also motivated by Kripke semantics~\cite{kripke_semantical_1963} (see \cite[Section 3]{Davies:POPL96} and
\cite[Section 4]{davies_modal_2001}) and is hence referred to as the Kripke style. 
% It has previously served
% as a foundation for simply typed multi-staged meta-programming \cite{davies_modal_2001}. 
%Similar to Fitch style, \lambox captures the Kripke
%semantics~\cite{kripke_semantical_1963}.
 Unlike Fitch-style systems where
 worlds are represented by segments between two adjacent ``lock''
 symbols, in $\lambox$, each world is represented by a context in a
 context stack. Nevertheless, the conversion between Kripke
and Fitch styles is largely straightforward\footnote{However, we note
  that \lambox has never been identified as or called a Fitch-style system.}. Here, we will often use ``context'' and ``world''
interchangeably. 
In \lambox, a term $t$ is typed in a context stack $\vGamma$ where
initially, the context stack consists of a single local context
which is itself empty (i.e. $\epsilon;\cdot$).  
\begin{mathpar}
  \mtyping[\epsilon; \Gamma_1; \ldots; \Gamma_n] t T

  \text{or}

  \mtyping t T
\end{mathpar}
%The declarations in the context are the assumptions we
%make in the corresponding world. 
The rightmost (or topmost) context represents the
current world.  In the $\square$ introduction rule, we extend the
context stack with a new world (i.e. new context). In the elimination
rule, if $\square T$ is true in a context stack
$\vGamma$, then $T$ is true in any worlds 
$\vGamma; \Delta_1;\ldots; \Delta_n$ reachable from
$\vGamma$. 
% $\unbox n t$ has type $T$ in a context stack
% $\vec{\Gamma}; \Delta_1; \ldots; \Delta_n$. 
The choice of the level
$n$ corresponds to reflexivity and transitivity of the accessibility
relation between worlds in the Kripke semantics.
%\cite{kripke_semantical_1963}.
\begin{mathpar}
  \inferrule*
  {\mtyping[\vGamma; \cdot] t T}
  {\mtyping{\boxit t}{\square T}}

  \inferrule*
  {\mtyping t {\square T}}
  {\mtyping[\vGamma; \Delta_1; \ldots; \Delta_n]{\unbox n t}{T}}    
\end{mathpar}
There are two key advantages of this $\tunbox$ formulation in \lambox.
First, it introduces a syntactic convenience to use natural numbers to describe levels
and therefore allows us to elegantly capture various modal logics
differing only in one parameter of the $\tunbox$ rule. By introducing $\tunbox$
levels, $\square$ is naturally non-idempotent.  Having $\tunbox$ allows us to study the relation of various
sublogics of $S4$ and treat them uniformly and compactly. % (see Table
% \ref{tab:systems}). 

%\begin{table}[htb]
% \vspace*{-10pt}
\begin{center}
  \begin{tabular}{|l|l|l|l|l|}
    \hline
    Axiom \textbackslash\ System & $K$          & $T$         & $K4$   & $S4$               \\ \hline
    $K$: $\square (S \func T) \to \square S \to \square T$ & $\checkmark$ & $\checkmark$ & $\checkmark$ & $\checkmark$ \\
    $T$: $\square T \to T$ &              & $\checkmark$ & &  $\checkmark$ \\
    $4$: $\square T \to \square \square T$ &              &              & $\checkmark$ & $\checkmark$ \\
    \hline
    $\tunbox$ level (UL) $n$ & \{ 1 \} & \{ 0, 1 \} & $\N^+$ & $\N$ \\
    \hline
  \end{tabular}  
\end{center}
%  \caption{Modal type theories, K, T, K4, S4} \label{tab:systems}
%  \vspace*{-10pt}
%\end{table}

Second, compared to dual-context formulation, it directly
corresponds to computational idioms quote ($\tbox$) and unquote
($\tunbox$) in practice, thereby giving a logical foundation to
multi-staged meta-programming \cite{davies_modal_2001}. In particular,
allowing $n=0$ gives us the power to
not only generate code, but also to run and evaluate code. 

A major stumbling block in reasoning about \lambox (see also
\cite{goubaultlarrecq:96}) is the fact that it is not obvious how to
define substitution properties for context stacks. This
prevents us from formulating an explicit substitution calculus for
\lambox which may serve as an efficient implementation. More importantly,
it also seems to be the bottleneck in developing normalization proofs
for \lambox that can be easily adapted to the various subsystems of S4.

% The root cause of this problem is 

% Previously, \cite{Pfenning95mfps,davies_modal_2001} defined modal
% structural properties for context stacks; this includes for example
% modal weakening and fusion on context stacks which allow us to splice 
% in new contexts into the stack and fuse two adjacent
% contexts into one. However, formulating these structural properties in
% terms of substitution on context stacks and more generally lifting
% substitution properties to context stacks has been difficult. 

% Our main contributions are two normalization by evaluation (NbE) algorithms for \lambox which we prove sound and complete.
% The first one (\Cref{sec:presheaf}) is based on a presheaf model and
% the second one (\Cref{sec:domain}) is
% extracted from an untyped domain  model adapting the work by
% \cite{abel_normalization_2013}. Central to the model
% constructions and the derived NbE algorithms is the concept of \emph{unified
%  (simultaneous) substitutions} on context stacks (\Cref{sec:usubst})
% which brings together two previously separate concepts: structural modal transformations on context stacks (such as
% modal weakening and fusion) and substitution properties for
% individual assumptions within a given context. This allows us to also formulate an explicit
% substitution calculus for \lambox.
% Our development and normalization proofs are generic and accommodate 
% all for subsystems of S4 \emph{without change}. More broadly, we see
% this work as a key step towards developing a
% dependently typed modal type theory. 

In this paper, we make the following contributions:
\begin{enumerate}
\item We introduce the concept of of \emph{Kripke-style substitutions} on
  context stacks (\Cref{sec:usubst}) which combines two previously separate
  concepts: modal transformations on context stacks (such as modal
  weakening and fusion) and substitution properties for individual assumptions within
  a given context.
\item We extend the standard presheaf model~\cite{altenkirch_categorical_1995} for
  simply typed $\lambda$-calculus and obtain a normalization by evaluation (NbE)
  algorithm for \lambox in \Cref{sec:presheaf}.  One critical feature of our development is that the algorithm
  and the proof accommodate all four subsystems of S4 \emph{without change}.
\item As opposed to Nanevski et al.~\cite{nanevski_contextual_2008}, we provide a
  contextual type formulation in $\lambox$ inspired by our notion of Kripke-style
  substitutions in \Cref{sec:contextual} which can serve as a construct for describing open code in a
  meta-programming setting. 
\end{enumerate}

This work opens the door to a substitution calculus and normalization of a dependently
typed modal type theory.  There are A partial formalization~\cite{artifact} of this work in
Agda~\cite{agda,norell_towards_2007} and an accompanying technical
report~\cite{hu_investigation_2022}.

%%%%%

% \subsection{Outline}
% In \Cref{sec:syntax}, we formally introduce the
% Kripke-style simply typed modal type theory (\lambox). We then define
% unified substitution and discuss its properties in
% \Cref{sec:usubst}. Based on these ideas, we first construct 
% a presheaf model in \Cref{sec:presheaf}  and an untyped domain
% model in \Cref{sec:domain}. From each of these models, we extract
% normalization algorithms.  Last, we compare our work to others' and
% conclude in \Cref{sec:related}. 

%%% Local Variables:
%%% mode: latex
%%% TeX-master: "main"
%%% End:

\begin{figure*}
  \raggedright
  \vspace{-5px}
  \fbox{$\mtyping t T$}\quad Term $t$ has type $T$ in context stack $\vGamma$
  \begin{mathpar}    
    \inferrule*
    {x : T \in \Gamma}
    {\mtyping[\vGamma; \Gamma] x T}

    \inferrule*
    {\mtyping[\vGamma; \cdot] t T}
    {\mtyping{\boxit t}{\square T}}

    \inferrule*
    {\mtyping t {\square T} \\ |\vDelta| = n}
    {\mtyping[\vGamma; \vDelta]{\unbox n t}{T}}
    
    \inferrule*
    {\mtyping[\vGamma;\Gamma, x : S]t T}
    {\mtyping[\vGamma;\Gamma]{\lambda x. t}{S \func T}}

    \inferrule*
    {\mtyping t {S \func T} \\ \mtyping s S}
    {\mtyping{t\ s}{T}}
  \end{mathpar}
  \vspace{-20px}
  
  \fbox{$\mtyequiv t{t'}T$}\quad Terms $t$ and $t'$ have type $T$ and are equivalent in
  context stack $\vGamma$
  \begin{align*}
    &\text{$\beta$ equivalence:}
    \qquad\qquad
    \inferrule
    {\mtyping[\vGamma; \cdot]{t}{T} \\
      |\vDelta| = n}
    {\mtyequiv[\vGamma; \vDelta]{\unbox{n}{(\boxit t)}}{t\{n / 0 \}}{T}}
    \qquad
    \inferrule
    {\mtyping[\vGamma;(\Gamma, x : S)]t T \\ \mtyping[\vGamma; \Gamma] s S}
    {\mtyequiv[\vGamma; \Gamma]{(\lambda x. t) s}{t[s/x]}{T}} \\
    &\text{$\eta$ equivalence:}
    \qquad\qquad\qquad
    \inferrule
    {\mtyping{t}{\square T}}
    {\mtyequiv{t}{\boxit{(\unbox 1 t)}}{\square T}}
    \qquad\qquad\qquad
    \inferrule
    {\mtyping{t}{S \func T}}
    {\mtyequiv t {\lambda x. (t\ x)}{S \func T}}
  \end{align*}
  \caption{Typing judgments and some chosen equivalence judgments}\label{fig:typing}
\end{figure*}

\section{Definition of \lambox}\label{sec:syntax}

In this section, we introduce the simply typed modal $\lambda$-calculus, \lambox, by Pfenning, Wong and Davies~\cite{Pfenning95mfps,davies_modal_2001} more
formally. We concentrate here on the fragment containing function
types $S \func T$, the necessity modality $\square
T$, and a base type $B$. 
\[
  \begin{array}{l@{~}l@{~}l}
  S, T  & := & \ B \sep \square T \sep S \func T \hfill \mbox{Types, \Typ} \\
  l, m, n & \multicolumn{2}{r}{\mbox{$\tunbox$ levels or offsets, $\mathbb{N}$}} \\
  x, y & \multicolumn{2}{r}{\mbox{Variables, $\Var$}} \\
  s, t, u &:=& x \sep \boxit t \sep \unbox n t \sep \lambda x. t \sep s\ t \hfill\qquad \mbox{Terms, $\Exp$} \\
  \Gamma, \Delta, \Phi &:= & \cdot \sep \Gamma, x : T \hfill \mbox{Contexts, \Ctx}\\
  \vGamma, \vDelta &:= &\ \epsilon \sep \vGamma; \Gamma \hfill \mbox{Context  stack, $\vect{\Ctx}$} \\
  w &:= & v \sep \boxit w \sep \lambda x. w  \hfill \mbox{Normal form, $\Nf$} \\
  v &:= & x \sep v\ w \sep \unbox n v \hfill \mbox{Neutral form, $\Ne$}
  \end{array}
\]
Following standard practice, we consider variables, applications, and $\tunbox$ neutral.
Functions, boxed terms and neutral terms are normal. Note that we
allow reductions under binders and inside boxed terms. As a
consequence, a function or a boxed term is normal, if their body is
normal. 

We define typing rules and type-directed equivalence between terms in
\Cref{fig:typing}.
% $\mtyping t T$ denotes that $t$ has type 
% $T$ in the context stack $\vGamma$ and similarly $\mtyequiv{t}{t'}{T}$ denotes that $t$ and $t'$ of type $T$ are
% equivalent in the context stack $\vGamma$.
 We only show the rules for $\beta$ and $\eta$
equivalence of terms, but the full set of rules can be found
in~\Cref{apx:equivalence}. 
% The rules for functions are standard. 
We use Barendregt's abstract naming and
$\alpha$ renaming to ensure that variables are unique with respect to
context stacks.  The variable rule asserts that one can only refer to
a variable in the current world (the topmost context). In a typing judgment, we require all context
stacks to be non-empty, so the topmost context must exist.

From Kripke semantics' point of view, the introduction rule for $\square$ says a term
of $\square T$ is just a term of $T$ in the next world. The elimination rule brings
$\square T$ from some previous world to the current world. This previous world is
determined by the level $n$. As mentioned earlier, the choice of $n$ determines which
logic the system corresponds to. % (see ~\Cref{tab:systems}).
$|\vDelta|$ counts the number of contexts in $\vDelta$.

% \subsection{Axioms}
%% BP: This is already known and can be found for example in
%% Davies/Pfenning's work. Is it useful to recap?
To illustrate, we recap how the axioms in \Cref{sec:intro} can be described in
\lambox. $K$ is defined by choosing $n=1$.  Axiom $T$ requires that $n=0$ and Axiom $4$
is only possible when $\tunbox$ levels (\RUL{s}) can be $>1$.
\[
  \begin{array}{llp{1cm}llp{1cm}ll}
    K &: \square (S \func T) \to \square S \to \square T & & T &: \square T \to T &  & A4 &: \square T \to \square \square T \\% \tag{Axiom $4$} \\
    K\ f\ x &:=  \boxit{((\unbox 1 f) (\unbox 1 x))} & & T\ x &:= \unbox 0 x  &  & A4\ x &:= \boxit{(\boxit{(\unbox 2 x)})}
  \end{array}
\]
The term equivalence rules are largely standard. In the $\eta$ rule
for $\square$, we restrict $\tunbox$  to level $1$.  In the $\beta$ rule for $\square$, we rely on the
modal transformation operation~\cite{davies_modal_2001}, written as $\{n/0\}$, which 
allows us to transform the term $t$ which is well-typed in the context
stack $\vGamma; \cdot$ to the context stack $\vGamma; \vDelta$. We abuse slightly notation and use $;$ for both extending a
context stack with a context and appending two context stacks. 
%We will describe these modal transformations in more detail
%we will discuss shortly.
We will discuss modal transformations more later in this section.

\subsection{Term Substitutions}

A term substitution simply replaces a variable $x$ with a term $s$
in a term $t$. It simply pushes the substitution inside the subterms of
$t$ and avoiding capture using renaming. Below, we simply restate the
ordinary term substitution lemma:  
\begin{lemma}[Term Substitution]$\;$\\
  If $\mtyping[\vGamma; (\Gamma, x{:}S, \Gamma'); \vDelta]t T$ and
  $\mtyping[\vGamma;(\Gamma, \Gamma')]s S$,
  then $\mtyping[\vGamma; (\Gamma, \Gamma'); \vDelta]{t[s/x]}T$.
\end{lemma}

\subsection{Modal Transformations (MoTs)}

In addition to the usual structural properties (weakening and
contraction) of individual contexts, \lambox also relies on structural
properties of context stacks, e.g. in the $\beta$ rule for
$\square$. In particular, we need to be able to weaken a
context stack $\vGamma; \vGamma'$ to $\vGamma; \vDelta; \vGamma'$ by
splicing in additional contexts $\vDelta$ (\emph{modal
  weakening}). \emph{Modal fusion} allows us to combine two adjacent
contexts in a context stack transforming a context stack
$\vGamma;\Gamma_0;\Gamma_1;\vDelta$ to a context stack $\vGamma;
(\Gamma_0,\Gamma_1); \vDelta$.

These modal transformations (\emph{MoTs}) require us to relabel the level $n$
associated with the $\tunbox$ eliminator. This is accomplished by the
operation $t\{n/l\}$. Assume that $t$ is well-typed
in a context stack $\vGamma$. If $n > 0$, then at position $l$ in the
stack (i.e. $\vGamma = \vGamma';\vDelta$ and $|\vDelta| = l$), we splice
in $n - 1$ additional contexts. If $n = 0$, then 
this can be interpreted as fusing the two adjacent contexts at position
$l$ in the stack $\vGamma$. 
% Intuitively, a MoT
%turns a term in one world to a term in another by recursively rearranging
%$\tunbox$ levels:
%\begin{align*}
\[
  \begin{array}{ll}
  x\{n/l\} &:= x \\
  \boxit t\{n/l\} &:= \boxit{(t\{n/l+1\})} \\[0.2em]
  \unbox m t \{n/l\} &:=
                       \begin{cases}
                         \unbox m {(t\{n/l-m\})} & \text{if $m \le l$} \\
                         \unbox{n + m - 1}t &\text{if $m > l$}
                       \end{cases} \\[0.2em]
  \lambda x. t \{n/l\} &:= \lambda x. (t\{n/l\}) \\
  s\ t\{n/l\} &:= (s\{n/l\})\ (t\{n/l\})
  \end{array}
\]
%\end{align*}
% where $n$ is subject to \RUL imposed by the system. 

In the $\tbox$ case, 
$l$ increases by one, as we extend the context stack by a new
world. In the $\tunbox$ case, we distinguish cases based on the $\tunbox$
level $m$.  If $m \le l$, then we simply rearrange the \RUL{s} recursively in $t$. If $m > l$, we only
need to adjust the \RUL and do not recurse on $t$. 
MoTs satisfy the following lemma:
\begin{lemma}[Structural Property of Context Stacks]\label{lem:mtran-typ}$\;$\\
  If $\mtyping[\vGamma;\Gamma_0;\Delta_0;\cdots;\Delta_l]{t}{T}$, 
  then $\mtyping[\vGamma;\Gamma_0; \cdots; (\Gamma_n,
  \Delta_0);\cdots;\Delta_l]{t\{n/l\}}{T}$.
\end{lemma}

We call the case where $n = 0$ \emph{modal fusion} or just
\emph{fusion}. % Note that for $K$ and $K4$ fusion is impossible.
% \footnote{Depending on \RUL, fusion might not be possible.}. 
Other cases  are \emph{modal weakening}. These names can be made sense of from the
following examples:
\begin{itemize}
\item When $n = l = 0$, the lemma states that if $\mtyping[\vGamma;\Gamma_0;\Delta_0]{t}{T}$, then
  $\mtyping[\vGamma;(\Gamma_0, \Delta_0)]{t\{0/0\}}{T}$. Notice that $\Gamma_0$ and $\Delta_0$
  in the premise are fused into one in the conclusion, hence ``modal fusion''.
\item When $n = 2$ and $l = 0$, the lemma states that if $\mtyping[\vGamma;\Gamma_0;\Delta_0]{t}{T}$, then
  $\mtyping[\vGamma;\Gamma_0; \Gamma_1; (\Gamma_2, \Delta_0)]{t\{2/0\}}{T}$. A new
  context $\Gamma_1$ is inserted into the stack and the topmost context is (locally)
  weakened by $\Gamma_2$, hence ``modal weakening''. 
\item In the $\beta$ rule for $\square$, a MoT $\{n/0\}$ is used to transform $t$ in
  context stack $\vGamma; \cdot$ to $\vGamma; \vDelta$. 
\item When $l > 0$, the leading $l$ contexts are skipped. If $n = 2$, $l = 1$ and
  $\mtyping[\vGamma;\Gamma_0;\Delta_0; \Delta_1]{t}{T}$, then
  $\mtyping[\vGamma;\Gamma_0; \Gamma_1; (\Gamma_2, \Delta_0);
  \Delta_1]{t\{2/1\}}{T}$. Here $\Delta_1$ is kept as is. 
\end{itemize}

%%% Local Variables:
%%% mode: latex
%%% TeX-master: "main"
%%% End:

\section{Kripke-style Substitutions}\label{sec:usubst}

Traditionally, we have viewed term substitutions and modal
transformations as two separate operations~\cite{davies_modal_2001}. This makes reasoning about
\lambox complex. For example, a composition of $n$ MoTs leads to up to $2^n$ cases in the
$\tunbox$ case. This becomes quickly
unwieldy. % This begs the question: 
How can we avoid such case analyses
by \emph{unifying} MoTs and term substitutions as one operation that transforms context stacks? --  %Recall that the $\tunbox$ case of MoTs has a case analysis. For a composition of $n$
%MoTs, due to potential $\tunbox$ level changes, the number of case analyses is in the order of
%$2^n$. Together with ordinary substitutions, reasoning about MoTs by brute force is virtually
%impossible. To resolve this issue, 
We will view context stacks as a category and a special \emph{unifying} group of simultaneous
substitutions as morphisms (denoted by $\To$). MoTs are then simply a special case of
these morphisms. Lemma \ref{lem:mtran-typ} suggests to view a MoT as a morphism:
\[
  \begin{array}{c}
 \{n/l\} : \vGamma;\Gamma_0; \cdots; (\Gamma_n,
  \Delta_0);\cdots;\Delta_l \To \vGamma;\Gamma_0;\Delta_0;\cdots;\Delta_l    
  \end{array}
\]
because $\{n/l\}$ moves $t$ from the codomain context stack to the domain
context stack. 
%Thus this observation calls for a \emph{unified} notion of substitutions
%that represents both MoTs and ordinary substitutions.  
If this group of substitutions
are closed under composition, then a category of context
stacks can be organized. 
% Some readers might think that unified
% substitutions can simply be described as lists of terms. However, this does not
% suffice, as we are missing the information about MoTs. Consider a morphism
% $\epsilon; \Gamma_1; \Gamma_2 \To \epsilon; \Delta_1; \Delta_2$. We can supply a list
% of terms typed in $\epsilon; \Gamma_1; \Gamma_2$ to replace variables in
% $\Delta_2$. But in which context stack should terms substituting variables in $\Delta_1$
% be well-typed in? We will answer this question by considering compositions of MoTs.

\subsection{Composing MoTs}

If MoTs are just special substitutions, then the composition of substitutions must also
compose MoTs. The following diagram is a composition of multiple MoTs, forming a morphism $\vGamma \To \vDelta$:
\begin{center}
  \vspace{-5px}
  \begin{tikzpicture}
    \matrix (m) [matrix of math nodes, row sep=10pt]
    {
      \vGamma := \epsilon; & \Gamma_0; & \Delta_0; & \Delta_1; & (\Gamma_1, \Gamma_2, \Gamma_3); &
      \Delta_2; & \Gamma_4
      \\
      \vDelta := \epsilon; & \Gamma_0; & & & \Gamma_1; \Gamma_2; \Gamma_3; &
      & \Gamma_4 \\
    };
    \draw[stealth-, double] ($(m-2-1) + (-0.35, 0.3)$)  -- ($(m-1-1) + (-0.35, -0.3)$);
    \draw[<-] (m-2-2)  -- (m-1-2);
    \draw[<-] ($(m-2-5) + (-0.6, 0.3)$) -- ($(m-1-5) + (-0.6, -0.3)$);
    \draw[<-] ($(m-2-5) + (-0.1, 0.3)$) -- ($(m-1-5) + (-0.1, -0.3)$);
    \draw[<-] ($(m-2-5) + (0.4, 0.3)$) -- ($(m-1-5) + (0.4, -0.3)$);
    \draw[<-] (m-2-7) -- (m-1-7);
  \end{tikzpicture}
  \vspace{-5px}
\end{center}
This composition contains both fusion ($\Gamma_1$, $\Gamma_2$, $\Gamma_3$) and
modal weakening ($\Delta_0$, $\Delta_1$, $\Delta_2$). The thin arrows
correspond to contexts in both stacks. Some of these arrows are local identities
($\Gamma_0$ and $\Gamma_4$). They happen between contexts that are affected by modal
weakenings. The rest are local weakenings ($\Gamma_1$, $\Gamma_2$ and $\Gamma_3$); in
this case, they are affected by modal fusions. The first observation is that the size of gaps
between adjacent thin arrows varies, because it is determined by different MoTs. Another
observation is that thin arrows do not have to be just local weakenings; if they
contain general terms, then we obtain a general
simultaneous substitution. Moreover, thanks to the thin arrows, we know exactly which context stack each term 
should be well-typed in.  Combining all information, we arrive at
the definition of \emph{Kripke-style
  substitutions}:
\begin{definition}
  A Kripke-style substitution $\vsigma$, or just \textbf{K-substitution}, between context stacks is defined as
\[
  \begin{array}{ll}
    \sigma, \delta &:= () \sep \sigma, t/x
    \hfill\mbox{(Local) substitutions, $\Subst$}\\
     \vsigma, \vdelta &:=  \snil \sigma \sep \sext \vsigma n \sigma\qquad
                            \hfill\mbox{K-substitutions, $\Substs$} 
  \end{array}
\]
  \begin{mathpar}
    \inferrule
    {\sigma : \vGamma \To \Gamma}
    {{\snil { \sigma}} : \vGamma \To \epsilon;\Gamma}

    \inferrule
    {\vsigma : \vGamma \To \vDelta \\ |\vGamma'| = n \\ \sigma : \vGamma;\vGamma' \To \Delta}
    {\sext \vsigma n \sigma : \vGamma;\vGamma' \To \vDelta;\Delta}
  \end{mathpar}

 where a local substitution $\sigma : \vGamma \To \Gamma$ is defined as a list of
 well-typed terms in $\vGamma$ for all bindings in $\Gamma$. 
%and is associated with an offset $n$ that captures the modal
%  weakening and fusion of context stacks respectively.
% corresponds to the level associated with $\tunbox$. 
\end{definition}

%The core idea is simple: a local substitution $\sigma$
%replaces variables in a context $\Gamma$ and \red{an offset $n$ specifies the
%correspondence of contexts}. 
Just as context stacks must be non-empty and consist of at least one
context, a K-substitution must have a topmost local substitution
written as $\snil \sigma$ in the base case. It provides a mapping for the context
stack $\epsilon;\Gamma$. We extend a K-substitution $\vsigma$ with
$\Uparrow^n\!\!\!\sigma$ where $n$ captures the offset due to a MoT and $\sigma$ is
the local substitution.  %
To illustrate, the morphism in the previous diagram can be represented as
$\varepsilon; \sext {\sext {\sext {\sext \id 3 \wk_1} 0 \wk_2} 0
  \wk_3} 2 \id$ where $\id$ is the local
identity substitution and
$\wk_i : \Gamma_0;\Delta_0;\Delta_1;(\Gamma_1,\Gamma_2,\Gamma_3) \To {\Gamma_i}$
are appropriate local weakenings. We break down this representation:
\begin{enumerate}
\item We start with $\snil\id : \epsilon ; \Gamma_0 \To \epsilon; \Gamma_0$.
\item We add an offset $3$ and a local weakening $\wk_1$, forming
  $\sext {\snil \id} 3 \wk_1 : 
   \epsilon ; \Gamma_0; \Delta_0;\Delta_1;(\Gamma_1,\Gamma_2,\Gamma_3) \To
  \epsilon; \Gamma_0; \Gamma_1$. The offset $3$ adds three contexts to
  the domain stack ($\Delta_0$, $\Delta_1$ and $\Gamma_1,\Gamma_2,\Gamma_3$). Local weakening
  $\wk_1$ extracts $\Gamma_1$ from $\Gamma_1,\Gamma_2,\Gamma_3$.
\item We extend the K-substitution to 
$\sext
   {\sext {\snil \id} 3  \wk_1}
    0  \wk_2 : \epsilon ; \Gamma_0; \Delta_0;\Delta_1;(\Gamma_1,\Gamma_2,\Gamma_3) \To
  \epsilon; \Gamma_0; \Gamma_1; \Gamma_2$. Since the offset associated
  with $\wk_2$ is $0$, no context is added to the domain stack. This effectively represents fusion. $\wk_2$ is
  similar to $\wk_1$. 
\item The rest of the K-substitution proceeds similarly. 
\end{enumerate}

% In general, a substitution contains as many local substitutions as contexts in its
% codomain stack. 

Subsequently, we may simply write
$\vsigma; \sigma$ instead of $\sext \vsigma 1 \sigma$. In particular,
we will write $\vsigma ; \wk$ instead of $\sext \vsigma 1 \wk$ and 
$\vsigma ; \id$ instead of $\sext \vsigma 1 \id$.  We often omit offsets
that are $1$ for readability.

\subsection{Representing MoTs}

Now we show that MoTs are a special case of K-substitutions. Let $l := |\vDelta|$.  We
define modal weakenings as
\[
  \begin{array}{ll}
  \{n+1/l\} &: \vGamma;\Gamma_1; \cdots; (\Gamma_{n + 1}, \Delta_0);\vDelta \To
              \vGamma;\Delta_0;\vDelta \\
  \{n+1/l\} &= \varepsilon; 
   \underbrace{\id; \cdots;\id}_{|\vGamma|};  \Uparrow^{n + 1} \wk;
    \underbrace{\id; \cdots; \id}_{|\vDelta|}
  \end{array}
\]
where the offset $n + 1$ on the right adds $\Gamma_1; \cdots; (\Gamma_{n + 1},
\Delta_0)$ to $\vGamma$ in the domain
stack and $\wk : \vGamma;\Gamma_1; \cdots; (\Gamma_{n + 1}, \Delta_0) \To {\Delta_0}$.

Fusion is also easily defined:
\[
  \begin{array}{ll}
  \{0/l\} &: \vGamma;(\Gamma_1, \Gamma_2); \vDelta \To \vGamma;\Gamma_1; \Gamma_2; \vDelta \\
  \{0/l\} &= \varepsilon; \underbrace{\id; \cdots; \id}_{|\vGamma|}; \wk_1; \Uparrow^0\! 
              \wk_2; \underbrace{\id; \cdots; \id}_{|\vDelta|}
  \end{array}
\]
where the offset $0$ associated with $\wk_2$ allows us to fuse $\Gamma_1$ and $\Gamma_2$, and
$\wk_i : \vGamma;(\Gamma_1, \Gamma_2) \To {\Gamma_i}$ for $i \in \{1, 2\}$.

\subsection{Operations on K-Substitutions}

We now show that K-substitutions are morphisms in a category of context stacks. In order to define
composition, we describe two essential operations: 1) \emph{truncation}
($\trunc \vsigma n$) drops $n$ topmost substitutions from a
K-substitution $\vsigma$ and 2)
\emph{truncation offset} ($\Ltotal \vsigma n$) computes the \emph{total} number of contexts that need to be dropped from the domain context stack, given that we truncate $\vsigma$ by $n$. 
%
%
%The $\Ltotal \vsigma n$ operation takes as input the unified substitution
%$\vsigma$ and $n$, the number of local substitutions to be
%dropped from $\vsigma$. 
It computes the sum of $n$ leading
offsets. Let 
$\vsigma := \sext{\sext{\vsigma'}{m_n}{\sigma_n} ; \ldots}{m_1}{\sigma_1}
$, then $\Ltotal {\vsigma}{n} = m_n + \ldots + m_1$ and $\trunc \vsigma n = \vsigma'$. For the operation to be meaningful, $n$ must be less than $|\Delta|$.
\[
\begin{array}{ll}
\multicolumn{2}{l}{\text{Truncation Offset}\quad
  \Ltotal {\_} {\_} : (\vGamma \To \vDelta) \to \N \to \N }\\
  \Ltotal \vsigma 0 &:= 0 \\
  \Ltotal {\sext \vsigma  n  \sigma} {1 + m} &:= n + \Ltotal \vsigma m
\end{array}
\]
%For $\L(\vsigma, n)$ to make sense, we require $n < |\vDelta|$ and $n$ subject to
%RUL. 
%Again, $n$ must satisfy $n < |\vDelta|$ for the operation to be
%meaningful. 
%
Truncation simply drops $n$ local
substitutions regardless of the offset that is associated with 
each local substitution. 
\[
  \begin{array}{ll}
\multicolumn{2}{l}{\text{Truncation}\quad 
  \trunc {\_} {\_} : (\vsigma : \vGamma \To \vDelta) \to  (n:\N) \to \trunc{\vGamma} {\Ltotal {\vsigma} n} \To
           \trunc {\vDelta} n}\\
  \trunc \vsigma 0 &:= \vsigma \\
  \trunc {(\sext \vsigma m  \sigma)} {1 + n} &:= \trunc {\vsigma} n
  \end{array}
\]
Similar to truncation of K-substitutions, we rely on truncation of contexts, written as $\trunc{\vGamma}{n}$ which simply drops $n$ contexts from the context stack $\vGamma$, i.e. if $\vGamma = \vGamma'; \Gamma_1; \ldots;\Gamma_n$, then $\trunc \vGamma n = \vGamma'$.
Note that $n$ must satisfy $n < |\vGamma|$, otherwise the operation
would not be meaningful. 
We emphasize that no further restrictions are placed on $n$ and hence
our definitions apply to any of the combinations of Axioms $K$, $T$ and $4$
described in the introduction. 
% $\trunc \vsigma n$ computes the tail of $\vsigma$ after removing $n$ leading local
% substitutions. 

% $\textsf{L}$ operation are defined for any $n$ 
%Since each local substitution substitutes a context in $\vDelta$, we
%know $\vDelta$ must also be truncated by $n$. However, how many contexts we should truncate
%from $\vGamma$ is not statically known, so it must be computed by $L$. 

\subsection{K-Substitution Operation}

We now can define the K-substitution operation as follows:
\[
  \begin{array}{ll}
  x[\snil \sigma] &:= \sigma(x)\qquad\hfill
                       \mbox{lookup $x$ in $\sigma$}\\
  x[\sext \vsigma k \sigma] &:= \sigma(x)\\
  (\boxit t)[\vsigma] &:= \boxit{t[\sext \vsigma 1 {()}]} \\
  (\unbox n t)[\vsigma] &:= \unbox{\Ltotal {\vsigma}{n}}{(t[\trunc \vsigma n])} \\
  (\lambda x. t) [\snil \sigma] &:= \lambda x. t[\snil{(\sigma, x/x)}] \\
  (\lambda x. t) [\sext \vsigma k \sigma] &:= \lambda x. t[\sext \vsigma k {(\sigma, x/x)}] \\
  (s\ t)[\vsigma] &:= s[\vsigma]\ u[\vsigma]
  \end{array}
\]

In the $\tbox$ case, the recursive call adds to $\vsigma$ an empty local
substitution. Note that the offset must be $1$, since we extend in the
box-introduction rule our context stack with a new empty context.

The $\tunbox$ case for K-substitutions incorporates MoTs.
Instead of distinguishing cases based on the unbox
level $n$, we use the truncation offset operation to re-compute the
\RUL and 
the recursive call $t[\trunc \vsigma n]$ continues with $\vsigma$ truncated, because
$t$ is typed in a shorter stack.

Due to the typing invariants, we know that $\Ltotal \vsigma n$ is indeed defined for
all valid \RUL $n$ in all our target systems.  This fact can be checked easily. In
System $K$, since $n = 1$ and all MoT offsets in $\vsigma$ are $1$, we have that
$\Ltotal \vsigma n = 1$. In System $T$ where $n \in \{0, 1\}$ and $\vsigma$ only
contains MoT offsets $0$ and $1$, we have $\Ltotal \vsigma n \in \{0, 1\}$. In System
$K4$ where $n \ge 1$, $\Ltotal \vsigma n$ cannot be $0$ and thus
$\Ltotal \vsigma n \ge 1$. In System $S4$, since $n \in \N$,
$\Ltotal \vsigma n \in \N$ naturally holds.

The following lemma shows that K-substitutions are indeed the proper notion we seek:
\begin{lemma}
  If $\mtyping t T$ and $\vsigma : \vDelta \To \vGamma$, then
  $\mtyping[\vDelta]{t[\vsigma]}T$. 
\end{lemma}

\subsection{Categorical Structure}

We are now ready to organize 
K-substitutions into a category. First we define the identity K-substitution:
\[
  \begin{array}{l@{\;}l}
  \vect{\id}_{\vGamma} &: \vGamma \To \vGamma \\
  \vect{\id}_{\vGamma} &:= \snil \id; \id; \cdots; \id
  \end{array}
\]
where $\id$'s are appropriate local identities. We again omit the offsets
when we extend K-substitutions with $\id$, since they are $1$. 
We also omit the subscript $\vGamma$ on $\vect{\id}$ for readability.% whenever it can be inferred.

Composition is defined in terms of the K-substitution operation:
\[
  \begin{array}{r@{\;\;}l}
  \_ \;\circ \_ \;&: \vGamma' \To \vGamma'' \to \vGamma \To \vGamma' \to \vGamma
                \To \vGamma'' \\
  (\snil \sigma) \circ \vdelta &:= \snil{(\sigma[\vdelta])} \\
  (\sext \vsigma n \sigma) \circ \vdelta &:= 
  \sext {(\vsigma \circ  {(\trunc \vdelta n)})}{\Ltotal \vdelta n}{(\sigma[\vdelta])}
  \end{array}
\]
where $\sigma[\vdelta]$ iteratively applies $\vdelta$ to all terms in
$\sigma$. In the recursive case, we continue with a truncated
K-substitution $\trunc \vdelta n$ and recompute the offset.

Verification of the categorical laws is then routine:
\begin{theorem}
  Context stacks and K-substitutions form a category with identities and composition
  defined as above. 
\end{theorem}

\subsection{Properties of Truncation and Truncation Offset}

Finally, we summarize some critical properties of truncation and truncation offset. Let $\vsigma : \vGamma' \To \vGamma$ and $\vdelta : \vGamma'' \To \vGamma'$: 
\begin{lemma}
  If $n < |\vGamma|$, then $\Ltotal \vsigma n  < |\vGamma'|$. 
\end{lemma}
\begin{lemma}[Distributivity of Addition] % \; \quad \\ 
  If $n + m < |\vGamma|$, then 
$\trunc \vsigma {(n + m)} = \trunc {(\trunc \vsigma n)} m$ 
and  $\Ltotal \vsigma {n + m} = \Ltotal {\vsigma} {n} + \Ltotal {\trunc \vsigma n}  m$. %%% bp: REVISE doesn't look good, since once | is notation and once it is an operation
\end{lemma}
\begin{lemma}[Distributivity of Composition]
  If $n < |\vGamma|$,    
% $\Ltotal \vsigma n = k$ \\ and
  then $\Ltotal {\vsigma \circ \vdelta} n = \Ltotal {\vdelta}{\Ltotal \vsigma n}$
  and $\trunc {(\vsigma \circ \vdelta)} n = (\trunc \vsigma n) \circ (\trunc {\vdelta} {\Ltotal \vsigma n })$.
\end{lemma}
These properties will be used in \Cref{sec:presheaf}. % These properties are complete in
% the sense that they are the sufficient condition for further investigation into \lambox.
Later we will define other instances of truncation and truncation offsets but all
these instances satisfy properties listed here. Therefore, these properties
sufficiently characterize an algebra of truncation and truncation offset.

%%% Local Variables:
%%% mode: latex
%%% TeX-master: "main"
%%% End:

\section{Normalization: A Presheaf Model}\label{sec:presheaf}

In this section, we present our NbE algorithm based on a presheaf model. Once we determine
the base category, the rest of the construction is largely
standard following Altenkirch et al.~\cite{altenkirch_categorical_1995} with minor
differences, which we will highlight. % We use the presheaf
% model to build conceptual intuitions which will also be useful in \Cref{sec:domain}.

To construct the presheaf model, we first determine the base category. Then we
interpret types, contexts and context stacks to presheaves and terms to natural
transformations. After that, we define two operations, reification and reflection, 
and use them to define the NbE algorithm. Last, we briefly discuss the completeness and
soundness proof. 
The algorithm is implemented in Agda~\cite{artifact}. % (without proofs). 

\subsection{Kripke-style Weakenings}

In the simply typed $\lambda$-calculus (STLC), the base category is the category
of weakenings. In \lambox, we must consider the effects of MoTs and we
will use the more general notion of \emph{Kripke-style
  weakenings} or K-weakenings which characterizes how a well-typed term in \lambox moves from one context
stack to another.
\begin{definition}
  A K-weakening $\vgamma : \vGamma \To_w \vDelta$ is:
  \[
    \vgamma := \varepsilon \sep q(\vgamma) \sep p(\vgamma) \sep \sext
    \vgamma n {} \qquad\hfill     \mbox{(K-weakenings)}
  \]
  \begin{mathpar}
    \inferrule
    { }
    {\varepsilon: \epsilon ; \cdot \To_w \epsilon ; \cdot}

    \inferrule
    {\vgamma : \vGamma; \Gamma \To_w \vDelta; \Delta}
    {q(\vgamma) : \vGamma; (\Gamma, x : T) \To_w \vDelta; (\Delta, x : T)}

    \inferrule
    {\vgamma : \vGamma; \Gamma \To_w \vDelta; \Delta}
    {p(\vgamma) : \vGamma; (\Gamma, x : T) \To_w \vDelta; \Delta}

    \inferrule
    {\vgamma : \vGamma \To_w \vDelta \\ |\vGamma'| = n}
    {\sext \vgamma n {~} : \vGamma; \vGamma' \To_w \vDelta; \cdot}
    (\text{the offset $n$ depends on \RUL})
  \end{mathpar}
\end{definition}
The $q$ constructor is the identity extension of the K-weakening
$\vgamma$, while $p$ accommodates weakening of an individual
context. These constructors are typical in the category of
weakenings~\cite[Definition 2]{altenkirch_categorical_1995}. To accommodate MoTs, we add to the category of 
weakenings the last rule which transforms a context stack. 
In the last rule, the offset $n$ is again parametric, subject to the \emph{same} \RUL
as the syntactic system,  and its choice
determines which modal logic the system corresponds to.
Note that we also write $\vect \id$ for the identity K-weakening.
Following our 
truncation and truncation offset operations for K-substitutions
in~\Cref{sec:usubst}, we can easily define these operations
together with composition also for K-weakenings. We omit these
definitions for brevity and we simply note that a truncated 
K-weakening remains a K-weakening. Now we obtain the base category: 
\begin{lemma}
  K-weakenings form a category $\WC$. 
\end{lemma}

\subsection{Presheaves}

The NbE proof is built on the presheaf category $\widehat \WC$ over $\WC$. $\widehat\WC$ has presheaves
$\WC^{op} \To \SetC$ as objects and natural transformations as morphisms. We know from
the Yoneda lemma that two presheaves $F$ and $G$ can form a presheaf exponential $F
\hfunc G$:
\[
  \begin{array}{ll}
  F \hfunc G &: \WC^{op} \To \SetC \\
  (F \hfunc G)_{\;\vGamma} &:= \forall \vDelta \To_w \vGamma. F_{\;\vDelta} \to G_{\;\vDelta}    
  \end{array}
\]
It is natural in $\vDelta$. As a convention, we use subscripts for both functorial applications and
natural transformation components. As in~\cite{altenkirch_categorical_1995}, presheaf exponentials model functions.
To model $\square$, we define 
\[
  \begin{array}{rl}
  \hsquare F & : \WC^{op} \To \SetC \\
  (\hsquare F)_{\;\vGamma} & := F_{\;\vGamma; \cdot}    
  \end{array}
\]
where $F$ is a presheaf. In Kripke semantics, $\hsquare$ takes $F$ to
the next world. Unlike presheaf exponentials which always exist regardless of the base
category, $\hsquare$ requires the base category to have the notion of ``the next
world''. This dependency in turn allows us to embed the Kripke structure of context
stacks into the base category, so that our presheaf model can stay a moderate
extension of the standard construction~\cite{altenkirch_categorical_1995}.

With this setup, we give the interpretations of types, contexts, and
context stacks in \Cref{fig:intp-presheaf}.  
The interpretation of the base type $B$ is the presheaf from context stacks to neutral
terms of type $B$. We write $\Ne\ T\ \vGamma$ for the set of
neutral terms of type $T$ in stack $\vGamma$. $\Ne\ T$ then is the presheaf
$\vGamma \mapsto \Ne\ T\
\vGamma$. $\Nf\ T\ \vGamma$ and $\Nf\ T$ are defined similarly.
The case $\intp{\square T} := \hsquare \intp{T}$ states that
semantically, a value of $\intp{\square T}$ is just a value of $\intp{T}$ in the next
world, which implicitly relies on
unified weakening's capability of expressing MoTs.
$\hat\top$ are $\hat\times$ are a chosen terminal object and products in $\widehat\WC$
and $*$ is \emph{the only element} in the chosen singleton set.

The interpretation of context stacks is more interesting. In the step case, $\vGamma;
\Gamma$ is interpreted as a product. To extract both part of $\vrho
\in \intp{\vGamma}_{\vDelta}$, we write $(\pi, \rho) := \vrho$ where $(n,\vrho') := \pi$. 
The first component, namely $\pi$, again consists of two parts:
1) the level $n$ satisfying $n < |\vDelta|$ which corresponds to the
MoTs that we support. We note that our definitions again apply to any
of the combinations of Axioms $K$, $T$ and $4$ depending on the choice of
$n$. 2) the recursive interpretation of $\vGamma$ in the truncated stack $\trunc
\vDelta n$ described by $\vrho'$. This stack truncation is necessary to interpret $\tunbox$. 

The second component, namely $\rho$, describes the interpretation of
the top-most context $\Gamma$. 
The fact that our interpretation of context stacks stores the level $n$
ultimately justifies the offsets stored in K-substitutions.

\begin{figure*}
  \begin{minipage}[t]{.2\textwidth}
    \begin{align*}
      \intp{\_} &: \Typ \to \WC^{op} \To \SetC \\
      \intp{B} &:= \Ne\ B \\
      \intp{\square T} &:= \hsquare \intp{T} \\
      \intp{S \func T} &:= \intp{S} \hfunc \intp{T}
    \end{align*}
  \end{minipage}
  \begin{minipage}[t]{.3\textwidth}
    \begin{align*}
      \intp{\_} &: \Ctx \to \WC^{op} \To \SetC \\
      \intp{\cdot} &:= \hat\top \\
      \intp{\Gamma, x : T} &:= \intp{\Gamma} \hat\times \intp{T}
    \end{align*}
  \end{minipage}
  \begin{minipage}[t]{.4\textwidth}
    \begin{align*}
      \intp{\_} &: \vect\Ctx \to \WC^{op} \To \SetC \\
      \intp{\epsilon; \Gamma} &:= \hat\top \hat \times \intp{\Gamma} \\
      \intp{\vGamma; \Gamma}_{\;\vDelta} &:= (\Sigma n <
                                         |\vDelta|. \intp{\vGamma}_{\;\trunc \vDelta n})
                                         \;\times\; \intp{\Gamma}_{\;\vDelta}
                                         \tag{where the offset $n$ depends on \RUL }
    \end{align*}
  \end{minipage}
  \caption{Interpretations of types, contexts and context stacks to presheaves}\label{fig:intp-presheaf}
\end{figure*}

\begin{lemma}[Functoriality]
  $\intp{T}$, $\intp{\Gamma}$ and $\intp{\vGamma}$ are presheaves. 
\end{lemma}
Functoriality means the interpretations also act on morphisms in $\WC$. Given $\vgamma
: \vGamma \To_w \vDelta$ and $a \in \intp{T}_{\;\vDelta}$, we write $a[\vgamma] \in
\intp{T}_{\;\vGamma}$. We intentionally overload the notation for applying K-substitutions
to draw a connection. This notation also applies for morphism actions of
$\intp{\Gamma}$ and $\intp{\vGamma}$.

\begin{figure*}
  \vspace{-10px}
  \[
    \begin{array}{l@{\;}l}
      \multicolumn{2}{l}{\mbox{Evaluation}\qquad       \intp{\_}\!\quad:  \mtyping t T \to \intp{\vGamma} \To \intp{T}}\\
      \multicolumn{2}{l}{\mbox{Expanded form}~     \intp{t}_{\;\vDelta} : \intp{\vGamma}_{\;\vDelta} {\;{\to}\;} \intp{T}_{\;\vDelta}}\\
      % \intp{\_} &:  \mtyping t T \to \intp{\vGamma} \To \intp{T} \\
      % \intp{t}_{\;\vDelta} &: \intp{\vGamma}_{\;\vDelta} \to \intp{T}_{\;\vDelta}
      % \hfill                           \text{expanded form}\\
      \intp{x}_{\;\vDelta} ((\_, \rho))
                  & := \rho(x) \hfill\mbox{lookup $x$ in $\rho$} \\
      \intp{\boxit t}_{\;\vDelta}(\vrho)
                  &:=  \intp{t}_{\;\vDelta; \cdot} (((1, \vrho), *)) \\
      \intp{\unbox n t}_{\vDelta}(\vrho)
                  &:= \intp{t}_{\;\trunc \vDelta m}(\trunc \vrho n)[\vect \id; \Uparrow^m]
                    \qquad \mbox{where $m := \Ltotal \vrho n$ and $\vect
                    \id;\Uparrow^m: \vDelta \To_w \trunc \vDelta m ; \cdot$} \\
      \intp{\lambda x. t}_{\;\vDelta}(\vrho)
                  &:= (\vgamma : \vDelta' \To_w \vDelta)(a) \mapsto
                    \intp{t}_{\vDelta'} ((\pi, (\rho, a)))
                    \hfill \text{where $(\pi, \rho) := \vrho[\vgamma]$} \\[0.2em]
      \intp{t\ s}_{\;\vDelta}(\vrho) & := \intp{t}_{\;\vDelta} (\vrho, \vect{\id}_{\;\vDelta}~,~
                                       \intp{s}_{\;\vDelta}(\vrho))     \\[0.5em]
      \multicolumn{2}{l}{\mbox{Reification}\quad \downarrow^T : \intp{T} \To \Nf\ T}\\
      \downarrow^B_{\vGamma}(a) &:= a \\[0.75em]
      \downarrow^{\square T}_{\vGamma}(a)
                     &:= \boxit \downarrow^T_{\vGamma; \cdot}(a)
                       \hfill \mbox{notice $a \in (\hsquare\intp{T})_{\vGamma} = \intp{T}_{\vGamma; \cdot}$}\\[0.75em]
      \downarrow^{S \func T}_{\vGamma; \Gamma}(a)
                     &:= \lambda x. \downarrow^T_{\vGamma; (\Gamma, x : S)}(a~(p(\vect \id)~,~{\uparrow^S_{\vGamma; (\Gamma, x : S)}\!(x)}))
                      \hfill \mbox{where $p(\vect \id) : \vGamma; \Gamma, x{:} S \To_w \vGamma; \Gamma$}\\[0.5em]
      \multicolumn{2}{l}{\mbox{Reflection} \quad \uparrow^T : \Ne\ T \To \intp{T}}\\
      \uparrow^B_{\vGamma}(v) &:= v \\[0.75em]
      \uparrow^{\square T}_{\vGamma}(v)
                     &:= \uparrow^T_{\vGamma; \cdot}({\unbox 1 v}) \\[0.75em]
      \uparrow^{S \func T}_{\vGamma; \Gamma}(v)
                     &:= (\vgamma : \vDelta \To_w \vGamma;\Gamma)(a)
                       \mapsto \uparrow^T_{\vDelta}(v[\vgamma]\ \downarrow^S_{\vDelta}(a))
    \end{array}
  \]
  \vspace{-15px}
  \caption{Evaluation, reification and reflection functions}\label{fig:intp-nat}
\end{figure*}

\subsection{Evaluation}

The interpretation of well-typed terms to natural transformations, or
\emph{evaluation} (see \Cref{fig:intp-nat}),
relies on truncation and the truncation offset. % $\intp{\vGamma}$
These operations are defined below and follow the same principles that lie behind the corresponding operations for syntactic K-substitutions.  
% We define truncation and truncation
% offset now on semantic substitutions in the presheaf model. 
\[
  \begin{array}{llp{.5cm}ll}
\multicolumn{2}{l}{\mbox{Truncation Offset}\; \Ltotal{\_}{\_} :
    \intp{\vGamma}_{\;\vDelta} \to \N \to \N }
    & & \multicolumn{2}{l}{\mbox{Truncation}\; \trunc {\_} {\_} : (\vrho :
    \intp{\vGamma}_{\;\vDelta})\; (n:\N) \to \intp{\trunc\vGamma n}_{\trunc \vDelta
      {\Ltotal \vrho  n}}} \\
  \Ltotal \vrho 0 &:= 0 & & \trunc \vrho 0 &:= \vrho \\
    \Ltotal {((n, \vrho),\rho)} {1 + m}&:= n + \Ltotal \vrho m
                                         & & \trunc {((n, \vrho), \rho)}{1 + m} &:= \trunc \vrho m
  \end{array}
\]
Most cases in evaluation are straightforward.  In the $\tbox$ case the recursion
continues with an extended environment and $t$ in the next world. In the $\tunbox$
case, we first recursively interpret $t$ with a truncated environment and then the
result is K-weakened. This is because from the well-typedness of $t$, we know
$\typing[\trunc\vGamma n]t{\square T}$, so
$\intp{t}_{\;\trunc \vDelta m}(\trunc \vrho n)$ gives an element in set
$\intp{\square T}_{\;\trunc \vDelta m} = \intp{T}_{\;\trunc \vDelta m; \cdot}$. To
obtain our goal $\intp{T}_{\;\vDelta}$, we can apply monotonicity of $\intp{T}$ using a
morphism $\vDelta \To_w \trunc\vDelta m; \cdot$, which is given by
$\vect \id; \Uparrow^m$.
The cases related to functions are identical to~\cite{altenkirch_categorical_1995}.
In the $\lambda$ case, since we need to return a set function due to presheaf
exponentials, we use $\mapsto$ to construct this function. We first K-weaken the
environment $\vrho$ and then extend it with the input value $a$.  % The left of $\mapsto$ is arguments and the
% right is the function body.
In the application case, since $\intp{t}$ gives us a presheaf exponential, we just
need to apply it to $\intp{s}$. We simply supply $\vect\id_{\vDelta}$ for the 
K-weakening argument because no extra weakening is needed. 

The following lemma proves that $\intp{t}$ is a natural transformation in $\vDelta$:
\begin{lemma}[Naturality]\label{lem:presh:nat}
  If $\mtyping t T$ and $\vrho \in \intp{\vGamma}_{\;\vDelta}$, then for all
  K-weakenings $\vgamma : \vDelta' \To_w \vDelta$, we have $\intp{t}_{\;\vDelta'}(\vrho[\vgamma]) =
  \intp{t}_{\;\vDelta}(\vrho)[\vgamma]$. 
\end{lemma}
The lemma states that the result of evaluation in a K-weakened environment is the same
as K-weakening the result evaluated in the original environment.
In STLC, despite being
a fact, naturality is not used anywhere in the proof. In \lambox, 
since K-weakenings encode MoTs, naturality is necessary 
in the completeness proof. 

After evaluation, we obtain a semantic value of the semantic type $\intp{T}$.  In the last
step, we use a \emph{reification} function to convert the semantic value back to a
normal form. Reification is defined mutually with \emph{reflection} in
\Cref{fig:intp-nat}. As suggested by their signatures, they are both natural
transformations, but our proof does not rely on this fact. Both reification and
reflection are type-directed, so after reification we obtain $\beta\eta$ normal
forms. We reify a semantic value $a$ of box type $\square T$ in a context stack
$\vGamma$ recursively extending the context stack to $\vGamma;\cdot$. Note that $a$
has the semantic type $(\hsquare\intp{T})_{\vGamma}$ which is defined as
$\intp{T}_{\vGamma; \cdot}$.  In the case of function type
$S \func T$, since $a$ is a presheaf exponential, we supply a 
K-weakening and a value, the result of which is then recursively reified.

Reflection turns neutral terms into semantic values. We reflect neutral terms of type
$\square T$ recursively and incrementally extending the context stack with one context
at a time. 
%\footnote{\color{red}{bp:how do you reduce $\tunbox~n(\boxit t)$?? -- it seems that the normal forms generated are loosing the ability of $\tunbox~n$. }}. 
% \begin{center}
% \color{red}{bp: Add explanation for neutral terms of function type}
% \end{center}
In the function case, to construct a presheaf exponential, we first take two
arguments $\vgamma$ and $a$. Since $v$ is a neutral term, $v[\vgamma]$ is also neutral but now
well-typed in $\vDelta$. Both recursive calls to reification and reflection then go
down to $\vDelta$ instead. 

Normalization by evaluation (NbE) takes a well-typed term $t$ in a context stack $\vGamma$ as input, interprets $t$ to its semantic counterpart in the initial environment, and reifies it back. Before defining NbE more formally, we define the identity environment $\intp{\vGamma}_{\vGamma}$ that is used as the initial environment:

\[
  \begin{array}{l@{}l}
  \uparrow &: (\vGamma : \vect\Ctx) \to \intp{\vGamma}_{\vGamma} \\
  \uparrow^{\epsilon; \cdot} &:= (*, *) \\
  \uparrow^{\vGamma; \cdot} &:= ((1, \uparrow^{\vGamma}), *) \\
  \uparrow^{\vGamma; \Gamma, x : T} &:= (\pi, (\rho, \uparrow^{T}_{\vGamma; \Gamma, x : T}\!\!(x)))\hfill
% & \multicolumn{1}{r}{\quad
    \;\;\mbox{where $(\pi, \rho) := \uparrow^{\vGamma; \Gamma} [p(\vect \id)]$}
  \end{array}
\]

Finally we define the NbE algorithm:
\begin{definition}(Normalization by Evaluation)
If $\mtyping t T$, then 
%  \begin{align*}
\[  
 \nbe^T_{\vGamma}(t) := \downarrow^T_{\vGamma} (\intp{t}_{\vGamma}(\uparrow^{\vGamma}))
\]
%  \end{align*}
\end{definition}

\subsection{Completeness and Soundness}

The algorithm given above is sound and complete:
\begin{theorem}
  [Completeness] If $\mtyequiv{t}{t'}{T}$, then $\nbe^T_{\vGamma}(t) = \nbe^T_{\vGamma}(t')$. 
\end{theorem}
\begin{theorem}
  [Soundness] If $\mtyping t T$, then $\mtyequiv{t}{\nbe^T_{\vGamma}(t)}T$. 
\end{theorem}

Due to space limitation, we are not able to present the whole proof. Fortunately, the
proof is very standard~\cite{altenkirch_categorical_1995}. To prove completeness, we simply need to prove that equivalent
terms always evaluate to the same natural transformation:
\begin{lemma}
  If $\mtyequiv{t}{t'}{T}$, then for any $\vrho \in \intp{\vGamma}_{\vDelta}$,
  $\intp{t}_{\vDelta}(\vrho) = \intp{t'}_{\vDelta}(\vrho)$.
\end{lemma}
\begin{proof}
  Induct on $\mtyequiv{t}{t'}{T}$ and apply naturality in most cases about $\square$.
\end{proof}

The soundness proof is established by a \emph{Kripke gluing model}. The gluing model
$t \sim a \in \glu{T}_{\;\vGamma}$ relates a syntactic term $t$ and a natural
transformation $a$, so that after reifying $a$, the resulting normal form is
equivalent to $t$:
\begin{align*}
  \glu{T}_{\;\vGamma} &\subseteq \Exp \times \intp{T}_{\;\vGamma} \\
  t \sim a \in \glu{B}_{\;\vGamma} &:= \mtyequiv{t}{a}{B} \\
  t \sim a \in \glu{\square T}_{\;\vGamma} &:= \mtyping{t}{\square T} \tand \forall
                                \vDelta. \unbox{|\vDelta|}{t} \sim a[\vect \id; \Uparrow^{|\vDelta|}] \in \glu{T}_{\vGamma; \vDelta} \\ 
  t \sim a \in \glu{S \func T}_{\;\vGamma} &:= \mtyping{t}{S \func T} \tand \forall
                                \vgamma : \vDelta \To_w \vGamma, s \sim b \in
                                \glu{S}_{\;\vDelta}. t[\vgamma]\ s \sim a(\vgamma, b) \in
                                \glu{T}_{\;\vDelta}
\end{align*}
The gluing model should be monotonic in $\vGamma$, hence Kripke. Again the gluing
model is very standard~\cite{altenkirch_categorical_1995}. It
is worth mentioning that in the $\square T$ case, the Kripke predicate effectively requires that $t$
and $a$ are related only when their results of \emph{any} $\tunbox$ing remains related.
We then can move on to prove some properties of the gluing model and define its
generalization to substitutions, which eventually allow us to conclude the soundness
theorem. Please find more details in our technical report~\cite{DBLP:journals/corr/abs-2206-07823}.

% Since soundness
% does not rely on naturality, we refer readers
% to~\cite{altenkirch_categorical_1995,abel_normalization_2013} or our technical
% report for details.

\subsection{Adaptiveness}
We emphasize that our construction is stable no matter our choice of
\RUL. Hence, our construction applies to all four modal systems, $K$, $T$,
$K4$ and $S4$ that we introduced in \Cref{sec:intro} \emph{without change}. The key insight
that allows us to keep our construction and model generic is the fact 
that K-substitutions,
K-weakenings, and $\intp{\vGamma}$ are instances of the algebra
formed by truncation and truncation offsets and
satisfy all the properties, in particular identity and 
distributivity, listed at the end of \Cref{sec:usubst}. 
% we don't have commutativity
More importantly, all the truncation and truncation offset functions 
are defined for all choices of \RUL thereby accommodating all four
modal systems with their varying level of $\tunbox$ing. 

%In the next section, we define another NbE algorithm with explicit substitutions using
%an untyped model. This algorithm continues to be generic and can
%therefore also easily be adapted for each modal system.

%%% Local Variables:
%%% mode: latex
%%% TeX-master: "main"
%%% End:

\section{Contextual Types}\label{sec:contextual}

In $S4$ and a meta-programming setting, $\square$ is interpreted as stages, where a term of
type $\square T$ is considered as a term of type $T$ but available only in the next
stage. However, as pointed out in~\cite{davies_modal_2001,nanevski_contextual_2008},
$\square$ only characterizes closed code. Nanevski et al.~\cite{nanevski_contextual_2008}
propose contextual types which relativize the surrounding context of a term so
representing open code becomes possible. However, this notion of contextual types
is in the dual-context style and how contextual types can be formulated with $\tunbox$ and
context stacks remains open. In this section, we answer this question
by utilizing our notion of K-substitutions. 

\subsection{Typing Judgments and Semi-K-substitutions}

With contextual types, we augment the syntax as follows:
\begin{align*}
  S, T &:= \cdots \sep \cbox{\vDelta}{T} &
  s,t,u &:= \cdots \sep \cbox{\vDelta}{t} \sep \cunbox{t}{\svsigma}
\end{align*}
$\cbox{\vDelta}{T}$ is a contextual type. It captures a list of contexts $\vDelta$ which a
term of type $T$ can be open in. Note that $\vDelta$ here can be
empty. This notion of contextual types is very general and captures a term open in
multiple stages. $\cbox{\vDelta}{t}$ is the constructor of a contextual type, where the contexts
that it captures are specified. $\cunbox{t}{\svsigma}$ is the eliminator. Instead of
an $\tunbox$ level, we now require a different argument $\svsigma$, which is a
\emph{semi-K-substitution} storing $\tunbox$ offsets and terms. We will discuss more
very shortly. 

The introduction rule for contextual types is straightforward:
\begin{mathpar}
  \inferrule
  {\mtyping[\vGamma;\vDelta]{t}{T}}
  {\mtyping{\cbox{\vDelta}{t}}{\cbox{\vDelta}{T}}}
\end{mathpar}
If we let $\vDelta = \epsilon; \cdot$, then we recover $\square$. If we let $\vDelta =
\epsilon; \Delta$ for some $\Delta$, then we have an open term $t$ which uses only
assumptions in the same stage. If $\vDelta$ has more contexts, then $t$ is an open
term which uses assumptions from previous stages. We can also let $\vDelta =
\epsilon$. In this case, $\cbox{\epsilon}{T}$ is isomorphic to $T$ and is not too
meaningful but allowing so makes our formulation mathematically cleaner. 

The elimination rule, on the other hand, becomes significantly more complex:
\begin{mathpar}
  \inferrule
  {\mtyping[\trunc\vGamma{\sLtotal{\svsigma}}]{t}{\cbox{\vDelta}T} \\
    \svsigma : \vGamma \To_s \vDelta}
  {\mtyping{\cunbox{t}{\svsigma}}{T}}
\end{mathpar}
It is no longer enough to eliminate with just an $\tunbox$ level because the
eliminator must specify how to replace all variables in $\vDelta$ and how
contexts in $\vGamma$ and $\vDelta$ relate. This information is collectively stored in a
\emph{semi-K-substitution} $\svsigma$ (notice the semi-arrow), which intuitively is not yet a valid
K-substitution, but close:
\begin{definition}
  A semi-K-substitution $\svsigma$ is defined as follows:
  \begin{align*}
    \svsigma, \svdelta := \varepsilon \sep \sext \svsigma n \sigma
    \tag*{Semi-K-substitutions, $\SSubsts$}
  \end{align*}
  \begin{mathpar}
    \inferrule
    { }
    {\varepsilon : \vGamma \To_s \epsilon}

    \inferrule
    {\svsigma : \vGamma \To_s \vDelta \\ |\vGamma'| = n \\ \sigma : \vGamma;\vGamma' \To \Delta}
    {\sext \svsigma n \sigma : \vGamma;\vGamma' \To_s \vDelta;\Delta}
  \end{mathpar}
\end{definition}
Compared to K-substitutions, semi-K-substitutions differ in the base case, where
empty $\varepsilon$ is permitted, so they are not valid K-substitutions. However, if a
semi-K-substitution is prepended by an identity K-substitution, then the result is a
valid K-substitution. Also, $\sLtotal{\svsigma}$ computes the sum of all offsets in $\svsigma$:
\[
  \begin{array}{llp{2cm}ll}
    \vect \id; &: \SSubsts \to \Substs & & \sLtotal{\_} &: \SSubsts \to \N \\ 
    \vect \id; \varepsilon &:= \vect\id & & \sLtotal \varepsilon &:= 0\\
    \vect \id; (\sext \svsigma n \sigma) &:= \sext{(\vect \id; \svsigma)}n\sigma
                                       & & \sLtotal{\sext \svsigma n \sigma} &:=
                                                                               \sLtotal
                                                                               \svsigma
                                                                               + n
  \end{array}
\]
We can prove the following lemma:
\begin{lemma}\label{lem:ssubsts-id}
  If $\svsigma : \vGamma \To_s \vDelta$, then $\vect\id; \svsigma : \vGamma \To_s
  \trunc\vGamma{\sLtotal{\svsigma}}; \vDelta$. 
\end{lemma}
This lemma is needed to justify the $\beta$ equivalence rule which we are about to
discuss.

\subsection{Equivalence of Contextual Types}

Having defined the introduction and elimination rules, we are ready to describe how
they interact. Note that the congruence rules are standard so we omit them here and
only describe the $\beta$ and $\eta$ rules:
\begin{mathpar}
  \inferrule
  {\mtyping[\trunc \vGamma {\sLtotal{\svsigma}};\vDelta]{t}{T} \\
    \svsigma : \vGamma \To_s \vDelta}
  {\mtyequiv{\cunbox{\cbox{\vDelta}{t}}{\svsigma}}{t[\vect\id; \svsigma]}{T}}

  \inferrule
  {\mtyping{t}{\cbox{\vDelta}{T}}}
  {\mtyequiv{t}{\cbox{\vDelta}{\cunbox{t}{\svect \id}}}{\cbox{\vDelta}{T}}}
\end{mathpar}
In the $\eta$ rule, $\svect \id$ denotes the identity semi-K-substitution, which is
defined as
\begin{align*}
  \svect \id_{\vDelta} &: \vGamma; \vDelta \To_s \vDelta \\
  \svect \id_{\vDelta} &:= \varepsilon; \underbrace{\id; \cdots ; \id}_{|\vDelta|}
\end{align*}
We omit the subscript whenever possible. Both rules are easily justified. In the $\beta$ rule, since $t$ is typed in the
context stack $\trunc \vGamma {\sLtotal{\svsigma}};\vDelta$, we obtain a term in
$\vGamma$ by applying $\vect\id; \svsigma$ due to Lemma \ref{lem:ssubsts-id}. In the
$\eta$ rule, by definition, we know
$\trunc{(\vGamma; \vDelta)}{\sLtotal{\svect \id}} = \vGamma$ and therefore
$\mtyping[\vGamma; \vDelta]{\cunbox{t}{\svect \id}}{T}$.

In an extensional setting, where the constructor and the eliminator of modalities are
congruent as done in this paper, we can show that the contextual type $\cbox{\epsilon;
\Delta_1; \cdots; \Delta_n}{T}$ is isomorphic to $\square(\Delta_1 \to \cdots \square
(\Delta_n \to T))$ if we view contexts $\Delta_i$ as iterative products. This implies
introducing contextual types does not increase the logical strength of the system and
the system remains normalizing. Nevertheless, contextual types given here seem to have
a natural adaptation to dependent types and set a stepping stone towards representing
open code with dependent types and therefore a homogeneous,
dependently typed meta-programming system.
%%% Local Variables:
%%% mode: latex
%%% TeX-master: "main"
%%% End:

\section{Related Work and Conclusion}\label{sec:related}

\subsection{Modal Type Theories}

There are many early attempts to give a constructive formulation of
modal logic, especially the modal logic S4 starting back in the 1990's~\cite{bierman_intuitionistic_1996,bierman_intuitionistic_2000,bellin_extended_2001,alechina_categorical_2001,borghuis_coming_1994,martini_computational_1996}.
% but they have problems
% handling substitutions. 
Pfenning and Davies~\cite{Davies:POPL96,pfenning_judgmental_2001} give the first formulation of
$S4$ in the dual-context style where we separate the assumptions that are
valid in every world from assumptions that are true in the current
world. This leads to a dual-context style formulation that
satisfies substitution properties and has found many applications 
from staged computation to homotopy type theory (HoTT). For example, 
Shulman~\cite{shulman_brouwers_2018} extends idempotent $S4$ with dependent types, called spatial type
theory and Licata et al.~\cite{licata_internal_2018} define crisp type theory,
which removes the idempotency from spatial type theory. However, both papers do not
give a rigorous justification of their type theories.
%The typing judgment keeps track of two kinds of contexts, one for valid assumptions
%($\Delta$) and the other for true assumptions ($\Gamma$):
%\begin{align*}
%  \typing[\Delta \sep \Gamma]t T
%\end{align*}
Most recently Kavvos~\cite{kavvos_dual-context_2017} investigates modal 
systems based on this dual-context formulation for Systems $K$, $T$, $K4$ and $S4$ as well
as the L\"ob induction principle. Kavvos also gives categorical semantics for 
these systems.  

However, it has been difficult to develop direct normalization proofs
for these dual-context formulations, since we must handle extensional properties like commuting
conversions (c.f. \cite{kavvos_dual-context_2017,girard_proofs_1989}). Further, our four target systems have very different formulations in
the dual-context style as shown by Kavvos~\cite{kavvos_dual-context_2017}. As a
consequence, it is challenging to have one single normalization
algorithm for all our four target systems.
%%% 
%  in
% contrast to Kripke style, where four systems are parameterized by RUL and thus allow
% adaptive normalization algorithms as we have shown. Prior to this paper, dual-context style has the
% advantage of a clear substitution structure (namely two simultaneous substitutions,
% one for the valid assumptions and one for the true assumptions). However, this point
% is no longer an advantage, as this
% paper has established simultaneous substitutions in Kripke style. We believe that
% using the idea of unified substitutions, we are able to provide a simpler operational
% semantics for quote-unquote metaprogramming systems~\cite{moggi_idealized_1999,taha_multi-stage_1997} and simpler subject reduction proofs.

An alternative to the dual-context style is the Fitch-style
approach pursued by Clouston,
Birkedal and collaborators (see
\cite{clouston_fitch-style_2018,gratzer_implementing_2019,birkedal_modal_2020}). At the
high-level, Fitch-style systems also model the Kripke semantics, but instead of using one context for each world,
the Fitch style uses a special symbol (usually $\thelock$) to segment one context
into multiple sections, each of them representing one world. Variables to the left of
the rightmost $\thelock$ are not accessible. 
Our normalization proof and the generalization of $\lambox$ to contextual types
also can likely be adapted to those systems. 

% The
% introduction rule for $\square$ in Fitch style is
% \begin{mathpar}
%   \inferrule
%   {\typing[\Gamma, \thelock]{t}{T}}
%   {\typing{\boxit t}{\square T}}
% \end{mathpar}
Clouston~\cite{clouston_fitch-style_2018} gives Systems $K$ and idempotent $S4$ in
the Fitch style and discusses their categorical
semantics. Gratzer et al.~\cite{gratzer_implementing_2019} describe idempotent $S4$ with
dependent types. 
% We put a detailed discussion of idempotent $S4$ in Fitch style in \Cref{sec:idem}.
Birkedal et al.~\cite{birkedal_modal_2020} give $K$ with dependent types and formulate dependent
right adjoints, an important categorical concept of
modalities. Gratzer et al.~\cite{gratzer_multimodal_2020,gratzer_multimodal_2021,gratzer_multimodal_2022}
proposes MTT, a multimode type theory, which describes interactions between multiple
modalities. Though MTT uses $\thelock$ to segment contexts, we believe
that MTT is better understood as a generalization of the dual-context
style and is apparent in the let-based formulation of the box
elimination rule. This different treatment of the box elimination also
 makes it less obvious how to understand $\lambox$ as a subsystem
 of MTT.

 Currently, existing Fitch-style systems mostly consider idempotent $S4$ where $\square T$ is isomorphic to
$\square\square T$. However, we consider this distinction to be
important from a computational view. For example, in multi-staged programming (see
\cite{pfenning_judgmental_2001,davies_modal_2001}) $\square T$
and $\square\square T$ describe code generated in one stage and two stages,
respectively. Moreover, $\unbox 0 t$ is
interpreted as evaluating and running the code generated by $t$.
% Notice that nonidempotent $S4$ is
It is nevertheless possible to develop a non-idempotent $S4$ system using $\tunbox$
levels $n$ in the Fitch style by defining a function which truncates a context until its
$n$'th $\thelock$. This is however more elegantly handled in \lambox, because
worlds are separated syntactically.  For this reason, we consider \lambox
as a more versatile and more suitable foundation for developing a
dependently typed meta-programming system. In particular, our
extension to contextual types shows how we can elegantly accomodate
reasoning about open code which is important in practice. 

Though context stacks in \lambox are taken from Pfenning, Wong and Davies'
development~\cite{pfenning_judgmental_2001,davies_modal_2001},
Borghuis~\cite{borghuis_coming_1994} also uses context stacks in his
development of modal pure type systems. The elimination rules use explicit weakening
and several ``transfer'' rules while \lambox incorporates both using $\tunbox$ levels,
which we consider more convenient and more practical from a programmer's
point of view. Martini and Masini~\cite{martini_computational_1996} also use context stacks. Their system
annotates all terms with a level which we consider too verbose to be practical.

\subsection{Normalization}

For the dual-context style, Nanevski et al.~\cite{nanevski_contextual_2008} give
contextual types and prove normalization by
reduction to another logical system with permutation
conversions~\cite{de_groote_strong_1999}. This means that the proof is indirect and
does not directly yield an algorithm for normalizing
terms. Kavvos~\cite{kavvos_dual-context_2017} gives a rewriting-based normalization
proof for dual-context style systems with L\"ob induction. Most
recently, Gratzer~\cite{gratzer_multimodal_2022} proves the normalization for MTT.  It is not
clear to us whether techniques in~\cite{gratzer_multimodal_2022} scale
to dependently typed Kripke-style systems, as the system have
different treatment of the box elimination. 

There are two recent papers closely related to our work: Valliappan et al.~\cite{VRC}
and Gratzer et al.~\cite{gratzer_implementing_2019}. %
\cite{VRC} gives different simply typed formulations in the Fitch style for all four
subsystems of $S4$ and as a result, a different normalization proof must be given to
each subsystem individually. %
Gratzer et al.~\cite{gratzer_implementing_2019} follow
Abel~\cite{abel_normalization_2013} and give an NbE proof for dependently typed
idempotent $S4$. Since the proof in~\cite{gratzer_implementing_2019} is parameterized
by an extra layer of poset to model the Kripke world structure introduced by
$\square$, as pointed out in~\cite{gratzer_multimodal_2020}, this proof cannot even be
easily adapted to dependently typed $K$ (see Birkedal et
al.~\cite{birkedal_modal_2020}). %
Compared to these two papers, our model is a moderate extension to the standard
presheaf model, requiring no such extra layer and adapting to multiple logics
automatically, and we are confident that it will generalize more easily to the
dependently typed setting.  The ultimate reason why we only need one proof to handle
all four subsystems of $S4$ is that we \emph{internalize} the Kripke structure of
context stacks in the presheaf model. The internalization happens in the base
category, where MoTs are encoded as part of K-weakenings. The internalization
captures peculiar behaviours of different systems and conflates the extra Kripke
structure from context stacks and the standard model construction, so that the proofs
become much simpler and closer to the typical construction.

\subsection{Conclusion and Future Work}

In this paper, we present a normalization-by-evaluation (NbE) algorithm
for the simply-typed modal $\lambda$-calculus (\lambox) which covers
all four subsystems of S4. The key to achieving this result is our
notion of K-substitutions which 
%carry more structure than typical
%simultaneous substitutions. In particular, they 
provides a unifying account for modal
transformations and term substitutions and allows us to formulate
a substitution calculus for modal logic S4 and its various
subsystems. Such calculus is not only important from a 
practical point of view, but play also a central role in our theoretical
analysis. Using insights gained from K-substitutions we organize a presheaf model, from which we extract a normalization
algorithm. The algorithm can be implemented in conventional programming languages
and directly account for the normalization of \lambox. Deriving from 
K-substitutions, we are also able to give a formulation for contextual types with
$\tunbox$ and context stacks, which had
been challenging prior to our observation of K-substitutions and is
important for representing open code in a meta-programming setting.

This work serves as a basis for further investigations into
coproducts~\cite{altenkirch_normalization_2001} and categorical structure of context
stacks.
%We emphasize that the method shown in this paper is general and can easily be customized
%to account for modal systems $T$, $K$, $K4$, and $S4$. 
We also see this
work as a step towards a Martin-L\"of-style modal type theory in which open code has
an internal shallow representation. With a dependently typed extension and contextual types,
it would allow us to
develop a \emph{homogeneous} meta-programming system with dependent
types which has been challenging to achieve. 

%Last, we expect that unified substitutions
%in other\footnote{what other? Diamond? Circle? -- I am not sure
%  whether I fully buy this.} modal systems can be achieved by adding
%extra structures. One possible future work is to define their unified
%substitutions for commonly seen modal systems and see if this concept helps simplify existing work.

% Another direction is to move on to dependent types. We have had an extension in
% Martin-L\"of style but we would like to know 

% Discovering their simultaneous
% substitutions is crucial in the process of their justifications. 

% In this paper, we show that by carrying more structures in substitutions, we obtain
% very easy models of modal type theories which lead to proofs of normalization. This
% observation has a great potential and we expect that this observation can be used to
% develop elegant substitution calculi for many other modal type theories, as we briefly
% scatched in \Cref{sec:idem}. 

% Another possible direction is to define a categorical semantics for unified
% substitutions. This has already been somewhat started by the formulation of explicit
% substitutions in \Cref{sec:domain}. Notice that across the whole paper, $L$ and
% truncation have been consistently playing important roles. A categorical framework
% would also help to describe these two operators and understand their importance.

% TODO: maybe more

%%% Local Variables:
%%% mode: latex
%%% TeX-master: "main"
%%% End:

\begin{ack}
  This work was funded by the Natural Sciences and Engineering Research Council of
  Canada (grant number 206263), Fonds de recherche du Québec - Nature et Technologies
  (grant number 253521), and Postgraduate Scholarship - Doctoral by the Natural
  Sciences and Engineering Research Council of Canada awarded to the first author.
\end{ack}

\bibliographystyle{entics}
\bibliography{ref}
% \interlinepenalty=10000
%\newcommand{\newblock}{}
%% \bibliography{ref}
%
%\input{main.bbl}

\appendix

\section{Equivalence Rules}\label{apx:equivalence}

\begin{mathpar}
  \inferrule*
  {\mtyequiv s t T}
  {\mtyequiv t s T}

  \inferrule*
  {\mtyequiv s t T \\ \mtyequiv t u T}
  {\mtyequiv s u T}
  
  \inferrule*
  {x : T \in \Gamma}
  {\mtyequiv[\vGamma; \Gamma] x x T}

  \inferrule*
  {\mtyequiv[\vGamma; \cdot]{t}{t'}T}
  {\mtyequiv{\boxit t}{\boxit t'}{\square T}}

  \inferrule*
  {\mtyequiv{t}{t'}{\square T} \\
    |\vDelta| = n}
  {\mtyequiv[\vGamma; \vDelta]{\unbox n t}{\unbox n t'}{T'}}  

  \inferrule*
  {\mtyequiv[\vGamma;(\Gamma, x : S)]{t}{t'}T}
  {\mtyequiv[\vGamma; \Gamma]{\lambda x. t}{\lambda x. t'}{S \func T}}

  \inferrule*
  {\mtyequiv{t}{t'}{S \func T} \\
    \mtyequiv{s}{s'}S}
  {\mtyequiv{t\ s}{t'\ s'}T}

  \inferrule*
  {\mtyping[\vGamma; \cdot]{t}{T} \\
    |\vDelta| = n}
  {\mtyequiv[\vGamma; \vDelta]{\unbox{n}{(\boxit t)}}{t\{n / 0 \}}{T}}

  \inferrule*
  {\mtyping[\vGamma;(\Gamma, x : S)]t T \\ \mtyping[\vGamma; \Gamma] s S}
  {\mtyequiv[\vGamma; \Gamma]{(\lambda x. t) s}{t[s/x]}{T}}

  \inferrule*
  {\mtyping{t}{\square T}}
  {\mtyequiv{t}{\boxit{\unbox 1 t}}{\square T}}
  
  \inferrule*
  {\mtyping{t}{S \func T}}
  {\mtyequiv t {\lambda x. (t\ x)}{S \func T}}
\end{mathpar}

\end{document}